\documentclass[a4paper,english,thm-restate,cleveref,nameinlink]{lipics-v2021}
\hideLIPIcs
\nolinenumbers

\usepackage{dsfont}
\usepackage{todonotes}

\usepackage{comment}

\DeclareTextFontCommand{\emph}{\color{teal}\em}

\newcommand{\N}{\ensuremath{\mathds{N}}}

\newcommand{\R}{\ensuremath{\mathds{R}}}
\newcommand{\Sp}{\ensuremath{\mathcal{S}}}

\newcommand{\dist}[2]{\ensuremath{\left\lVert#1-#2\right\rVert}}
\newcommand{\geo}[2]{\ensuremath{\mathrm{d}_{\gamma}(#1,#2)}}

\graphicspath{{figures/}}

\title{Geometric Realizations 
of
Dichotomous Ordinal Graphs} 

\titlerunning{Geometric Realizations of Dichotomous Ordinal Graphs} 

\author{Patrizio Angelini}{John Cabot University, Rome, Italy}{pangelini@johncabot.edu}{https://orcid.org/0000-0002-7602-1524}{}
\author{Sabine Cornelsen}{University of Konstanz, Germany}{sabine.cornelsen@uni-konstanz.de}{https://orcid.org/0000-0002-1688-394X}{}
\author{Carolina Haase}{Trier University, Germany}{haasec@uni-trier.de}{https://orcid.org/0000-0001-6696-074X}{}
\author{Michael Hoffmann}{Department of Computer Science, ETH Z\"urich, Switzerland \and\url{https://people.inf.ethz.ch/hoffmann/}}{hoffmann@inf.ethz.ch}{https://orcid.org/0000-0001-5307-7106}{}
\author{Eleni Katsanou}{National Technical University of Athens, Greece}{ekatsanou@mail.ntua.gr}{https://orcid.org/0000-0002-1001-1411}{}
\author{Fabrizio Montecchiani}{University of Perugia, Italy}{fabrizio.montecchiani@unipg.it}{https://orcid.org/0000-0002-0543-8912}{}
\author{Raphael Steiner}{ETH Z{\"u}rich, Switzerland}{raphaelmario.steiner@inf.ethz.ch}{https://orcid.org/0000-0002-4234-6136}{}
\author{Antonios Symvonis}{National Technical University of Athens, Greece}{symvonis@math.ntua.gr}{https://orcid.org/0000-0002-0280-741X}{}

\authorrunning{Angelini, Cornelsen, Haase, Hoffmann, Katsanou, Montecchiani, Steiner, Symvonis} 

\Copyright{Patrizio Angelini, Sabine Cornelsen, Carolina Haase, Michael Hoffmann, Eleni Katsanou,\\
Fabrizio Montecchiani, Raphael Steiner, and Antonios Symvonis} 

\ccsdesc[500]{Mathematics of computing~Combinatorics}
\ccsdesc[500]{Mathematics of computing~Graph theory}
\ccsdesc[500]{Human-centered computing~Graph drawings}

\keywords{Ordinal embeddings, geometric graphs, graph representations} 

\acknowledgements{This work was initiated at the Annual Workshop on Graph and Network Visualization (GNV~2023), Chania, Greece, June 2023.} 

\EventEditors{Oswin Aichholzer and Haitao Wang}
\EventNoEds{2}
\EventLongTitle{41st International Symposium on Computational Geometry (SoCG 2025)}
\EventShortTitle{SoCG 2025}
\EventAcronym{SoCG}
\EventYear{2025}
\EventDate{June 23--27, 2025}
\EventLocation{Kanazawa, Japan}
\EventLogo{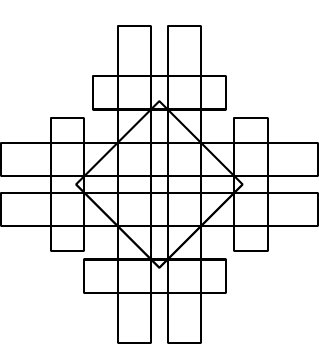}
\SeriesVolume{332}
\ArticleNo{XX}     

\begin{document}

\maketitle

\begin{abstract}
A \emph{dichotomous ordinal graph} consists of an undirected graph with a partition of the edges into \emph{short} and \emph{long} edges. 
A \emph{geometric realization} of a dichotomous ordinal graph~$G$ in a metric space~$X$ is a drawing of~$G$ in~$X$ in which every long edge is strictly longer than every short edge.
We call a graph~$G$ \emph{pandichotomous} in~$X$ if~$G$ admits a geometric realization in~$X$ for every partition of its edge set into short and long edges. 

We exhibit a very close relationship between the degeneracy of a graph~$G$ and its \emph{pandichotomic Euclidean or spherical dimension}, that is,  
 the smallest dimension~$k$ such that~$G$ is pandichotomous in~$\R^k$ or the sphere~$\Sp^k$, respectively. First, every~$d$-degenerate graph is pandichotomous in~$\R^{d}$ and~$\Sp^{d-1}$ and these bounds are tight for the sphere and for~$\R^2$ and almost tight for~$\R^d$, for~$d\ge 3$. Second, \emph{every} $n$-vertex graph that is pandichotomous in~$\R^k$ has at most $\mu kn$ edges, for some absolute constant $\mu<7.23$. This shows that the pandichotomic Euclidean dimension of any graph is linearly tied to its degeneracy and in the special cases $k\in \{1,2\}$ resolves open problems posed by Alam, Kobourov, Pupyrev, and Toeniskoetter.

Further, we characterize which complete bipartite graphs are pandichotomous in~$\R^2$: 
These are exactly the~$K_{m,n}$ with~$m\le 3$ or $m=4$ and~$n\le 6$.
For general bipartite graphs, we can guarantee realizations in~$\R^2$ if the short or the long subgraph is  constrained: namely if the short subgraph is outerplanar or a subgraph of a rectangular grid, or if the long subgraph forms a caterpillar.
 \end{abstract}

\section{Introduction}

For an \emph{ordinal embedding}, we are given a set of objects~$x_1,\dots,x_n$ in an abstract~space, together with a set of \emph{ordinal constraints} of the form $\mathrm{dist}(x_i,x_j) < \mathrm{dist}(x_k,x_\ell)$. The objective is to compute a set of points~$p_1,\dots,p_n$ in a metric space 
while preserving as many ordinal constraints as possible.
Ordinal embeddings were first studied in the 1960's by Shepard~\cite{Shepard1962,Shepard1962-2} and Kruskal~\cite{kruskal-msogfnh-64,Kruskal1964} in the context of psychometric data analysis. Recently, there have been applications in the field of Machine Learning~\cite{DBLP:conf/icml/TeradaL14}. The computation of ordinal embeddings is also known in the literature as \emph{non-metric multi-dimensional scaling.} For an extensive literature review on ordinal embeddings refer to
~\cite{VankadaraLHWL23}.

Of particular interest in relation to our work is the application of ordinal embeddings to the problem of recognizing 
multidimensional preferences~\cite{Bennett1960,DBLP:journals/jal/DoignonF94,peters:aaai17} in the field of Computational Social Science.
The objects are either \emph{voters} or \emph{options}, which, together with the ordinal constraints (i.e., the voters' preferences), naturally define a bipartite graph. However, 
the goal is to find an embedding in $\R^d$ where all constraints are satisfied rather than to seek for an approximation. 
Efficient algorithms exist when $d=1$~\cite{DBLP:journals/scw/ChenPW17,DBLP:journals/jal/DoignonF94}, while for any $d \geq 2$ the problem is as hard as the existential theory of the reals~\cite{peters:aaai17}. 
In the simplest case, the preferences form a \emph{dichotomy}, that is, a voter may either support or reject an option~\cite{DBLP:conf/ijcai/ElkindL15,peters:aaai17}.


This setting is naturally modeled as a \emph{dichotomous ordinal graph}, which
consists of an undirected graph~$G=(V, E_s \cup E_\ell)$ together with a partition of the edges into a set~$E_s$ of \emph{short} edges and a set $E_\ell$ of \emph{long} edges. 
A \emph{geometric realization} of a dichotomous ordinal graph~$G$ is a drawing~$\Gamma$ of~$G$ in some metric space along with a threshold $\delta > 0$ such that the short edges of~$G$ are exactly those that have length at most~$\delta$ in~$\Gamma$. In this work we consider two natural classes of drawings in which edges are drawn as geometrically shortest paths: (1)~straight-line drawings in the Euclidean space~$\R^d$ and (2)~geodesic drawings on the sphere~$\Sp^d$. \cref{fig:intro} shows two straight-line drawings of the same dichotomous ordinal graph in~$\R^2$. The drawing in~(a) is a valid geometric realization, whereas the drawing in~(b) is not.

\begin{figure}[htbp]
    \centering
    \begin{minipage}[t]{.45\textwidth}
    \centering
    \includegraphics[page=1]{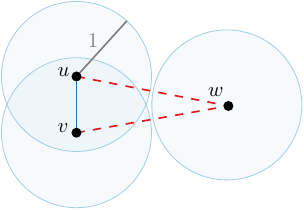}
    \subcaption{valid drawing of $\triangle uvw$}
    \end{minipage}
    \hfill
    \begin{minipage}[t]{.45\textwidth}
    \centering
    \includegraphics[page=2]{intro}
    \subcaption{invalid drawing of $\triangle uvw$}
    \end{minipage}
	\caption{
    A dichotomous ordinal triangle. 
    Short edges are shown blue/solid, long edges red/dashed.}    
	\label{fig:intro}
\end{figure}

A dichotomous ordinal graph may or may not be realizable in any particular space. For instance, it is easy to see that~$K_4$ where the short edges induce a star cannot be realized in~$\R^1$.
In general, 
it is  NP-hard to decide whether a dichotomous ordinal graph admits a geometric realization in~$\R^2$, even if the underlying graph is a complete graph and the short edges induce a planar graph~\cite[Lemma~1]{alam_etal:wg15}. 
According to Peters~\cite[Theorem~4]{peters:aaai17} the problem is $\exists\R$-complete for bipartite graphs.\footnote{But even the full version of the paper~\cite{p-rmep-16} does not contain an explicit proof but only refers to the work of Kang and M{\"u}ller~\cite[Theorem~1]{km-sdprg-12}. The graph constructed in this proof is not bipartite. While it is plausible that the construction could be adapted, it does not seem immediately obvious to us.}

In light of the motivation, it is desirable that a geometric realization always exists, no matter the preferences of the voters. We call a 
graph~$G=(V,E)$ \emph{pandichotomous} in a metric space~$X$ if~$G$ admits a geometric realization in~$X$ for every partition of~$E$ into short and long edges. 
Graphs that are pandichotomous in $\mathds R^1$ are also called total-threshold-colorable~\cite{alamCFKKPT:dam17,alamKPT:FUN14} 
or total weak unit interval graphs \cite{alam_etal:wg15}. 
Angelini, Bekos, Gronemann, and Symvonis~\cite{AngeliniBGS19} exhibited some graph classes that are pandichotomous in~$\R^2$, such as \emph{double-wheels} (a cycle and two additional vertices connected to all vertices of the cycle),
\emph{2-degenerate} 
graphs (can be reduced to the empty graph by repeatedly removing a vertex of degree at most two),
and \emph{subcubic} graphs 
(vertex degree at most three).
On the negative side, 
there exists a dichotomous ordinal graph~$G$ whose underlying graph can be obtained from a double-wheel by adding a single edge
such that~$G$ does not admit a geometric realization in~$\R^2$~\cite{AngeliniBGS19}. 
Clearly, being pandichotomous is a \emph{monotone} graph property, that is, if a graph~$G$ is pandichotomous in~$X$, then so is every subgraph of~$G$. For~$X=\R^d$, being pandichotomous is also closed under taking disjoint unions, but for~$X=\Sp^d$ this is not immediate.

A related question is the existence of a realization of a graph as a \emph{unit disk graph}, where vertices are represented by unit disks, and they are connected by an edge if and only if the corresponding disks intersect. More generally, the \emph{sphericity}~\cite{h-cdgacmc-82,m-sgs-84} of a graph~$G$ is the smallest~$d$ such that~$G$ can be realized as an intersection graph of unit balls in~$\R^d$. The main difference between geometric realizations of dichotomous ordinal graphs and unit ball realizations lies in the different types of edges. In a unit disk realization there are only two types: edge and non-edge, and all of them have to be faithfully represented. In a geometric realization of dichotomous ordinal graphs there are three types of edges: long, short, and non-edges, and we have no constraints concerning the last type. Therefore, only upper bounds on the sphericity carry over to the dichotomous setting. 
Since 
every graph on~$n$ vertices has sphericity at most~$n$~\cite{m-sgs-84}, 
every graph on~$n$ vertices is also pandichotomous in~$\R^n$. So for every finite graph~$G$ we can define its \emph{pandichotomic Euclidean dimension}~$\mathrm{ped}(G)$ to be the smallest dimension~$d$ such that~$G$ is pandichotomous in~$\R^d$. Analogously, we define the \emph{pandichotomic spherical dimension}~$\mathrm{psd}(G)$. Trivially, $\mathrm{ped}(G)\le \mathrm{psd}(G)+1$.

\subparagraph{Results.} In this paper, we initiate the study of the pandichotomic dimension of graphs, with a focus on the bipartite case. In \cref{sec:completeBipartite}, we characterize the complete bipartite graphs that are pandichotomous in~$\R^2$. Specifically, we show that~$K_{m,n}$ is pandichotomous in~$\R^2$ if and only if either (1)~$m\le 3$ or (2)~$m=4$ and~$n\le 6$.


A dichotomous ordinal graph has two natural induced subgraphs: The \emph{short} subgraph, 
induced by the short edges, and the \emph{long} subgraph, 
induced by the long edges. In \cref{sec:outerplanar}, we show that if either of these subgraphs is sufficiently constrained, then a realization in~$\R^2$ always exists. Specifically, this is the case if (1)~the graph is bipartite and the short subgraph is outerplanar (\cref{thm:bipartiteOuter}); (2)~the short subgraph forms a subgraph of a rectangular grid (\cref{thm:grid}); or (3)~the long subgraph forms a caterpillar (\cref{thm:long}). In (1) and (2), we can also ensure that the short subgraph is realized without edge crossings. However, there are bipartite dichotomous ordinal graphs that do not admit a geometric realization in~$\R^2$, even though the short subgraph is planar (see, e.g., \cref{thm:k47}). 

In \cref{sec:3gen}, we study the pandichotomic Euclidean dimension and show that it is very closely related to the degeneracy. 
A graph is~\emph{$k$-degenerate} if it can be constructed starting from the empty graph by iteratively applying the following operation: Add a new vertex and make it adjacent to at most~$k$ existing vertices. The \emph{degeneracy}~$\mathrm{d}(G)$ of a graph~$G$ is the smallest~$k$ such that~$G$ is~$k$-degenerate.

We show that all~$d$-degenerate graphs are pandichotomous on~$\Sp^{d-1}$ and in~$\R^d$, for~$d\ge 2$ (\cref{thm:3genr3}). In particular, it follows that all bipartite planar graphs are pandichotomous on~$\Sp^{2}$ and in~$\R^3$ (\cref{cor:planar-on-sphere}). Our bounds imply~$\mathrm{ped}(G)\le\mathrm{d}(G)$ and~$\mathrm{psd}(G)\le\mathrm{d}(G)-1$, for every graph~$G$ with~$\mathrm{d}(G)\ge 2$. We also show that these bounds are tight for the sphere (\cref{thm:3genr3-ub-s}) and for~$\R^2$ (\cref{thm:3deg2}) and almost tight for~$\R^d$, for~$d\ge 3$ (\cref{thm:3genr3-ub}). 

We also
show that \emph{every} $n$-vertex graph that is pandichotomous in~$\R^{d}$ has at most~$\mu dn$ edges, for some 
constant~$\mu<7.23$ (\cref{thm:universalindspace}), and this bound is optimal up to the value of $\mu$. In the special cases $d\in \{1,2\}$, this affirmatively answers two open problems posed explicitly by Alam, Kobourov, Pupyrev, and Toeniskoetter~\cite{alam_etal:wg15}. Consequently, $\mathrm{d}(G)/(2\mu)\le\mathrm{ped}(G)\le\mathrm{d}(G)$ and~$\mathrm{d}(G)/(2\mu)-1\le\mathrm{psd}(G)\le\mathrm{d}(G)-1$, for every graph~$G$ (\cref{cor:cdbound}). In other words, the pandichotomic Euclidean and spherical dimensions are linearly tied to the degeneracy. 
As a direct consequence of these bounds (\cref{cor:3genr3-ub}), we determine up to a constant factor the smallest dimension $d$ for which every $n$-vertex (bipartite) dichotomous ordinal graph is realizable in $\R^d$ (or $\mathcal{S}^{d-1}$). 

\section{Preliminaries}

For two points~$p,q\in\R^d$ we denote by~$\dist pq$ the Euclidean distance between~$p$ and~$q$. For a point~$c\in\R^d$ and a positive real number~$r$, the \emph{ball} $\mathrm{B}(c,r)$ of radius~$r$ around~$c$ is the set~$\{p\in\R^d\colon\dist pc\le r\}$ of all points that have Euclidean distance at most~$r$ to~$c$. For a set~$A\subset\R^d$ we denote by~$\partial A$ the \emph{boundary} of~$A$. The boundary of~$\mathrm{B}(c,r)$ is formed by the \emph{sphere} $\mathrm{S}(c,r)=\{p\in\R^d\colon\dist pc=r\}$. 
A \emph{unit} ball or sphere has a radius of~$r=1$. 
The \emph{geodesic distance}~$\geo pq\in[0,2\pi]$ between two points~$p,q\in\Sp_d$ is determined by the central angle of a shortest great circle arc between~$p$ and~$q$.

Given a finite set~$X$ of geometric objects, such as hyperplanes or spheres, in~$\R^d$, the \emph{arrangement}~$\mathcal{A}(X)$ of~$X$ is the partition of~$\R^d$ induced by~$X$ into so-called cells of various dimension. The maximal connected open subsets of~$\R^d\setminus\bigcup_{R\in X}R$ are the \emph{$d$-cells} of~$\mathcal{A}(X)$. And every relatively open~$k$-dimensional intersection of two or more~$d$-cells forms a~$k$-cell of~$\mathcal{A}(X)$. The~$d$-cells are also called \emph{full-dimensional cells} or even just \emph{cells} of~$\mathcal{A}(X)$. The~$0$-cells and $1$-cells are also called \emph{vertices} and \emph{edges} of~$\mathcal{A}(X)$, respectively. 

For geometric realizations of dichotomous ordinal graphs in~$\R^d$ we may fix a global scale and assume without loss of generality 
a threshold of~$\delta=1$. 
Furthermore, as we consider finite graphs only, we may assume that no two vertices have distance exactly one. 

\begin{observation}\label{obs:delta}
  If a dichotomous ordinal graph admits a geometric realization in~$\R^d$, then it admits a realization in~$\R^d$ where no two vertices have distance exactly one.
\end{observation}

In contrast, such a global rescaling does not work in general for geometric realizations on~$\Sp^d$. 
%
If the short subgraph of a dichotomous ordinal graph has several connected components, then we can draw these components mutually far apart.
This yields the following.

\begin{observation}\label{obs:connected}
  A dichotomous ordinal graph admits a geometric realization in~$\R^d$ if and only if each subgraph induced by a connected component of its short subgraph does.
\end{observation}

\section{Complete Bipartite Graphs}\label{sec:completeBipartite}

In this section, we prove
the following theorem. The proof is split into \cref{thm:complete,thm:k47,thm:k55}, in combination with the fact that being pandichotomous is a monotone graph property.

\begin{theorem}\label{thm:kmnr2}
  The complete bipartite graph~$K_{m,n}$ is pandichotomous in~$\R^2$ if and only if either (1)~$m\le 3$ or (2)~$m=4$ and~$n\le 6$.
\end{theorem}

A convenient way to reason about geometric realizations for bipartite graphs is in terms of arrangements of spheres. Consider a bipartite dichotomous ordinal graph~$G=(U \cup W,E)$ and suppose that the vertices of~$U$ are already drawn as points in~$\R^d$. Then, to obtain a geometric realization for~$G$ the task is to place each~$w \in W$ such that for each~$u\in U$ with~$uw\in E$ we have~$w\in\mathrm{B}(u,1)$ if and only if the edge~$uw$ is short; see~\cref{SUBFIG:K3m}. 

Let~$U=\{u_1,\ldots,u_n\}$, let~$D_i=\mathrm{B}(u_i,1)$,  let~$C_i=\partial D_i$, and let~$\mathcal{C}$ denote the arrangement of~$C_1,\ldots,C_n$.
To every vertex~$w \in W$ we associate a set~$V(w)\subseteq U$ such that~$u\in V(w)$ if and only if~$uw$ is a short edge in~$G$. We refer to~$V(w)$ as a \emph{singleton}, a \emph{pair}, or a \emph{triple} if~$|V(w)|=1$, $2$, or~$3$,
respectively. A subset~$X\subseteq U$ is \emph{realized} by a drawing of~$U$ if there is a cell~$r$ in~$\mathcal{C}$ such that~$r\subseteq D_i$ if and only if~$u_i\in X$. Then~$G$ admits a geometric realization if and only if there exists a drawing 
of~$U$ where~$V(w)$ is realized, for all~$w\in W$.


\begin{lemma}\label{thm:complete}
  The complete bipartite graph~$K_{3,m}$ is pandichotomous in~$\R^2$, for all~$m\in\N$. 
  The complete bipartite graph~$K_{4,m}$ is pandichotomous in~$\R^2$, for all~$1\le m\le 6$. 
\end{lemma}
\begin{proof}
  For~$K_{3,m}$ we can draw~$U=\{u_1,u_2,u_3\}$ so that all eight subsets of~$U$ are realized; see~\cref{SUBFIG:K3m}. Therefore, any dichotomous ordinal $K_{3,m}$ admits a geometric realization.
  For~$|U|\ge 4$ such a universal placement is not possible because an arrangement of~$n$ circles has at most~$n(n-1)+2$ cells~\cite{s-gter-26}. 
  So an arrangement of four circles has at most~$14$ cells, whereas a four-element set has~$16$ subsets.
  However, for~$|U|=4$ and~$|W|\le 6$ we can always obtain a geometric realization as follows. Let~$V(W)=\{V(w)\colon w\in W\}\subset 2^U$.
  
  If there are at least three pairs in~$V(W)$, then, given that~$|V(W)|\le|W|\le 6$, the number of triples plus the number of singletons in~$V(W)$ together is at most three. Thus, as~$|U|=4$, there exists a vertex~$u\in U$ such that~$\{u\}\notin V(W)$ and~$U\setminus\{u\}\notin V(W)$. So we can use the drawing depicted in~\cref{SUBFIG:K4m_pairs}, where we assign~$u$ to the central disk. As all subsets of~$U$ other than~$\{u\}$ and~$U\setminus\{u\}$ are realized, this is a valid geometric realization of~$G$.
  
  Otherwise, there are at most two pairs in~$V(W)$. We use the drawing depicted in~\cref{SUBFIG:K4m_triples}, where we assign the vertices of~$U$ to the disks so that both pairs in~$V(W)$ appear consecutively in the circular order of disks. (This works regardless of whether or not these pairs share a vertex.) As all subsets of~$U$ other than the two pairs that correspond to opposite circles in the drawing are realized, this is a valid geometric realization of~$G$.
\end{proof}

\begin{figure}[htbp]
    \centering
    \begin{minipage}[t]{.24\textwidth}
    \centering
    \includegraphics[scale=0.8, page=1]{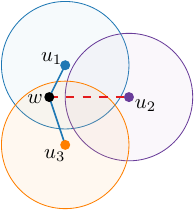}
    \subcaption{$K_{3,m}$\label{SUBFIG:K3m}}
    \end{minipage}
    \hfill
    \begin{minipage}[t]{.24\textwidth}
    \centering
    \includegraphics[scale=0.8, page=3]{complete_bipartite_disks}
    \subcaption{$K_{4,m}$ all pairs\label{SUBFIG:K4m_pairs}}
    \end{minipage}
    \hfill
    \begin{minipage}[t]{.24\textwidth}
    \centering
    \includegraphics[scale=0.8, page=2]{complete_bipartite_disks}
    \subcaption{$K_{4,m}$ all triples\label{SUBFIG:K4m_triples}}
    \end{minipage}
    \hfill
    \begin{minipage}[t]{.24\textwidth}
    \centering
    \includegraphics[scale=.8]{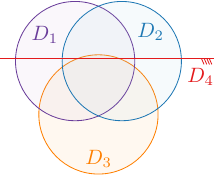}
    \subcaption{arbitrary radii\label{SUBFIG:realize}}
    \end{minipage}
    \caption{Arrangements of circles to represent one color class of a bipartite graph in~$\R^2$.} 
    \label{fig:completebipartite}
\end{figure}

The following elementary lemma about unit disks turns out helpful.
\begin{lemma}
\label{lem:chord}
   Let~$D,E$ be unit disks in~$\R^2$, let~$c$ be a chord of~$D$, and let~$A$ be a closed part of~$D$ bounded by~$c$ and~$\partial D$ whose interior~$A^\circ$ does not contain the center of~$D$. Then~$c\subset E\iff A\subset E$.
\end{lemma}
\begin{proof}\label{lem:chord:proof}
   The~$\Leftarrow$ implication is trivial because~$c\subset A$. It remains to prove the~$\Rightarrow$ implication. So suppose that~$c=pq\subset E$. If~$E=D$, then the statement is trivial. Hence suppose that~$E\ne D$. In particular, it follows that~$c$ is not diametrical (if it is, then it uniquely determines~$E=D$, and the statement holds for both parts bounded by~$c$ and~$\partial D$). 
   
   We obtain the disk~$E'\subseteq E$ by shrinking~$E$ concentrically until at least one of~$p,q$ lies in~$\partial E'$. Without loss of generality we have~$p\in\partial E'$. Next we obtain the disk~$E''\subseteq E'$ by shrinking~$E'$ such that~$E''$ is tangent to~$E'$ at~$p$ by moving    
   the center of the disk 
   along the ray towards~$p$ until~$q\in\partial E''$. 
   See \cref{SUBFIG:shrinkingE}.
   Note that both~$D$ and~$E''$ are disks whose bounding circle passes through~$p$ and~$q$. Therefore, both the center~$\mathrm{c}(D)$ of~$D$ and the center~$\mathrm{c}(E'')$ of~$E''$ lie on the orthogonal bisector of~$p$ and~$q$. Furthermore, the disk~$D$ has unit radius, whereas the radius of~$E''$ is at most one. Therefore, the distance between~$\mathrm{c}(D)$ and~$pq$ is at least as large as the distance between~$\mathrm{c}(E'')$ and~$pq$. Let~$A''$ denote the cap of~$E''$ induced by~$pq$, that is, the closed part of~$E''$ on the side of~$pq$ that does not contain~$\mathrm{c}(E'')$. As the bounding circles of~$D$ and~$E''$ intersect in exactly the two points~$p$ and~$q$, we have either~$A\subseteq A''$ or~$A''\subseteq A$. As~$\mathrm{c}(E'')$ is at least as close to~$pq$ as~$\mathrm{c}(D)$, we are in the former case and have~$A\subseteq A''$. Noting that~$A''\subset E''\subseteq E$ completes the proof.
   %
\end{proof}

\begin{figure}[htbp]
    \centering
    \begin{minipage}[t]{.48\textwidth}
    \centering
    \includegraphics[page=2]{proof2}
    \subcaption{Shrinking $E$\label{SUBFIG:shrinkingE}}
    \end{minipage}
    \hfill
    \begin{minipage}[t]{.48\textwidth}
    \centering
    \includegraphics[page=1]{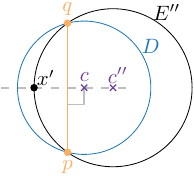}
    \subcaption{$E''$ is too large\label{SUBFIG:tooLarge}}
    \end{minipage}
    \caption{Illustration of the proof of \cref{lem:chord}.}
    \label{fig:chord}
\end{figure}

\begin{lemma}
\label{thm:k47}
   There exists a dichotomous ordinal~$K_{4,7}$ that is not realizable in~$\R^2$. 
\end{lemma}
\begin{proof}\label{thm:k47:proof}
   First we specify a dichotomous ordinal~$K_{4,7}$. Let~$U=\{u_1,u_2,u_3,u_4\}$ and~$W=\{w_1,\ldots,w_7\}$ denote the partition of the vertex set. For each~$w_i$, with~$1\le i\le 7$, we can specify its associated set~$U_i\subseteq U$ (such that exactly the edges between~$w_i$ and~$U_i$ are short). We choose all four subsets of size three and the three subsets of size two that contain~$u_4$, and distribute them among the vertices of~$W$ arbitrarily.

   Now consider an arbitrary but fixed geometric realization of this dichotomous ordinal graph, and the induced arrangement~$\mathcal{C}$. 
   In addition to the seven cells that correspond to the subsets that we specified explicitly, two more cells are required implicitly: As the disks are bounded, we always have an outer cell, which corresponds to the empty set, and there is also another cell that corresponds to the full set~$U$. The latter cell is required by Helly's Theorem~\cite{helly:23} because disks are convex and we specified all triples to be realized in~$\mathcal{C}$. 

   Consider the sub-arrangement~$\mathcal{C}^-$ that is induced by~$C_1,C_2,C_3$ in~$\mathcal{C}$. We have~$D_1\cap D_2\cap D_3\ne\emptyset$ because we require all triple intersections. Moreover, none of the three pairwise intersections is contained in the third disk because we require the triple intersection with~$D_4$. Therefore, combinatorially~$\mathcal{C}^-$ looks as depicted in \cref{SUBFIG:K3m}. By Helly the circle~$C_4$ crosses the cell~$f=D_1\cap D_2\cap D_3$ of~$\mathcal{C}^-$ but it does not fully contain it because we require the triple intersection~$D_1\cap D_2\cap D_3$ to form a cell in~$\mathcal{C}$. We claim that at least one vertex of~$\partial f$ lies outside of~$D_4$. 

   To prove the claim, observe that~$C_4$ is a unit circle and~$\partial f$ is formed by arcs of three unit circles. Suppose for a contradiction that~$D_4$ contains all three vertices~$\varphi_1,\varphi_2,\varphi_3$ of~$\partial f$. Then, as~$D_4$ is convex, the whole triangle~$\varphi_1\varphi_2\varphi_3$ is contained in~$D_4$. Applying \cref{lem:chord} to~$D_4$ and each of the edges of~$\varphi_1\varphi_2\varphi_3$ it follows that~$D_4\supset f$, a contradiction to our assumption that the triple intersection~$D_1\cap D_2\cap D_3$ forms a cell in~$\mathcal{C}$. This proves our claim.
   
   So without loss of generality we may suppose that the vertex of~$\partial f$ in~$C_1\cap C_2$ lies outside of~$D_4$.
   If~$D_4$ contains at most one vertex of~$\partial f$, say, at most the vertex in~$C_1\cap C_3$, then~$D_4$ does not intersect~$(D_1\cap D_3)\setminus f$, in contradiction to our requirement that the triple intersection~$D_1,D_3,D_4$ is realized in~$\mathcal{C}$. Thus, we conclude that~$D_4$ contains exactly two vertices of~$\partial f$, say, without loss of generality, the vertices on~$C_3$. But then~$D_4$ does not intersect~$D_3\setminus(D_1\cup D_2)$, a contradiction to our requirement that~$D_3\cap D_4$ must be realized in~$\mathcal{C}$. 
   Therefore, there is no geometric realization for this dichotomous ordinal~$K_{4,7}$.

   Note that it is essential that all disks in the realization must have the same radius. If the disks may have arbitrary radii, then we can obtain a realization of the example discussed above; see \cref{SUBFIG:realize}.
%
%
\end{proof}

\begin{lemma}\label{thm:k55}
   There exists a dichotomous ordinal~$K_{5,5}$ that is not realizable in~$\R^2$. 
\end{lemma}
\begin{proof}
   First we specify a dichotomous ordinal~$K_{5,5}$, see \cref{fig:k55counter} for an illustration. Let~$U=\{u_1,u_2,u_3,u_4,u_5\}$ and~$W=\{w_1,w_2,w_3,w_4,w_5\}$ denote the partition of the vertex set. To each~$w_i\in W$, for~$1\le i\le 5$, we associate a set~$\alpha(w_i)\subseteq U$ as follows:
   \[
   \alpha(w_i)=\{u_i,u_{i\oplus 1},u_5\}, \mbox{for~$1\le i\le 4$, and~}\alpha(w_5)=U\setminus\{u_5\},
   \]
   where~$i\oplus 1=(i\,\mathop{\mathrm{mod}}\,4)+1$. Now consider an arbitrary but fixed geometric realization of this dichotomous ordinal~$K_{5,5}$. In a slight abuse of notation we identify the vertices with the corresponding points in the geometric realization. Denote by~$D_i$ the unit disk centered at~$u_i$, which represents the region of points that are in short distance to~$u_i$. Let~$W^-=\{w_1,w_2,w_3,w_4\}$. As~$W^-\subset D_5$, whereas~$w_5\notin D_5$, by convexity of~$D_5$ all points of~$W^-$ lie in an open halfplane through~$w_5$. Let~$\phi_1,\phi_2,\phi_3,\phi_4$ denote the counterclockwise order of the points from~$W^-$ around~$w_5$ within this halfplane. By symmetry of~$W^-$ we may assume that~$\phi_1=w_1$ without loss of generality. We consider two cases.

\begin{figure}[htbp]
    \centering
    \begin{minipage}[t]{.48\textwidth}
    \centering
    \includegraphics[page=2]{complete_bipartite_short}
    \subcaption{short edges of $K_{4,7}$\label{fig:k55counter:1}}
    \end{minipage}
    \hfill
    \begin{minipage}[t]{.48\textwidth}
    \centering
    \includegraphics[page=1]{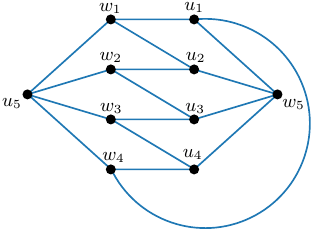}
    \subcaption{short edges of $K_{5,5}$}
    \end{minipage}
    \caption{A dichotomous ordinal~$K_{4,7}$ and $K_{5,5}$, respectively, that does not admit a geometric realization. The drawn edges are the short edges. Edges between vertices labeled $u$ on one hand and $w$ on the other hand, 
    that are not drawn, are long.}
    \label{fig:k55counter}
\end{figure}

   \emph{Case~1: The set~$\alpha(\phi_4)$ contains one of~$u_1$ or~$u_2$.} Then by symmetry between~$u_1$ and~$u_2$ we may assume without loss of generality that~$\phi_4=w_4$. See \cref{fig:circular}~(left) for illustration. By convexity of~$D_1$, the triangle~$\triangle=w_1w_4w_5=\phi_1\phi_4w_5$ is contained in~$D_1$. As~$w_2,w_3\notin D_1$, we have~$\phi_2,\phi_3\notin\triangle$. As~$w_2,w_3,w_5\in D_3$ but~$w_1,w_4\notin D_3$, the circle~$\partial D_3$ crosses the line segment~$w_1w_4$ twice in~$D_5$, and it also crosses both of the line segments~$\phi_1\phi_2$ and~$\phi_3\phi_4$. Similarly, as~$\triangle\subset D_1$ but~$w_2,w_3\notin D_1$, the circle~$\partial D_1$ crosses both line segments~$w_5w_2$ and~$w_5w_3$, and it also crosses both of the line segments~$\phi_1\phi_2$ and~$\phi_3\phi_4$. Let~$c$ denote the chord of~$D_3$ induced by~$\phi_1\phi_4$. As~$D_1\supset c$, by \cref{lem:chord} we know that~$D_1$ contains the part~$A$ of~$D_3$ on the side of~$c$ that does not contain the center of~$D_3$. This is    
   the part that contains~$w_5$ because the other part contains~$\phi_2,\phi_3\notin D_1$. Similarly, as~$D_5\supset c$, by \cref{lem:chord} it follows that also~$D_5\supset A$, which is a contradiction to~$w_5\notin D_5$. Thus, this case is impossible.

\begin{figure}[htbp]
    \centering%
    \includegraphics[page=1]{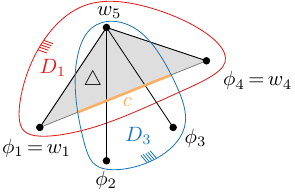}\hfil
    \includegraphics[page=2]{circularorder}
    \caption{The two cases in the proof of \cref{thm:k55}.}
    \label{fig:circular}
\end{figure}

   \emph{Case~2: The set~$\alpha(\phi_4)$ does not contain any of~$u_1$ or~$u_2$.} Then we have $\phi_4=w_3$ and both~$|\alpha(\phi_1)\cap\alpha(\phi_3)|=1$ and~$|\alpha(\phi_2)\cap\alpha(\phi_4)|=1$. By symmetry we may assume without loss of generality that~$\phi_2=w_2$ and~$\phi_3=w_4$. See \cref{fig:circular}~(right) for illustration. Consider the two disks~$D_1$ and~$D_3$. They both contain~$w_5$, and so the corresponding circles~$\partial D_1$ and~$\partial D_3$ intersect all of the rays~$w_5\phi_i$, for~$1\le i\le 4$. As~$\phi_1,\phi_3\in D_1\setminus D_3$ and~$\phi_2,\phi_4\in D_3\setminus D_1$, the rays~$w_5\phi_1$ and~$w_5\phi_3$ intersect~$\partial D_1$ before~$\partial D_3$, whereas the rays~$w_5\phi_2$ and~$w_5\phi_4$ intersect~$\partial D_3$ before~$\partial D_1$.
   But then~$\partial D_1$ and~$\partial D_3$ cross in each of the cones~$\phi_iw_5\phi_{i+1}$, for~$1\le i\le 3$, which is impossible, since any two distinct circles intersect in at most two points.

   As we arrived at a contradiction in both cases, we conclude that no geometric realization of this dichotomous ordinal~$K_{5,5}$ exists. If, however, we allow disks with arbitrary radii, 
   then a geometric realization exists, as depicted in \cref{fig:realizek55}.
\end{proof}

\begin{figure}[htbp]
    \centering%
    \includegraphics[page=2]{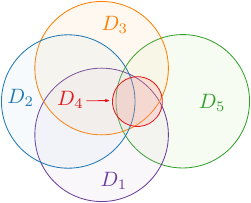}
    \caption{A realization of the dichotomous ordinal~$K_{5,5}$ from 
    \cref{thm:k55} by disks of arbitrary radii.}
    \label{fig:realizek55}
\end{figure}

\section{Graphs with Constrained Short or Long Subgraphs}
\label{sec:outerplanar}
We 
show that every bipartite dichotomous ordinal graph admits a geometric realization if the short subgraph~$G_s$ 
is outerplanar or a subgraph of the rectangular grid or if the  long subgraph~$G_\ell$ 
is a caterpillar. In the first case, we construct a plane drawing of $G_s$ in which the BFS-layers are drawn on horizontal lines (\cref{SUBFIG:outerOuter}). In the second case, we suitably perturb the grid (\cref{fig:grid}). And in the third case, we suitably place points on~$\Sp^2$.

\begin{figure}
	\centering
    \begin{minipage}[t]{.48\textwidth}
    \centering
    \includegraphics[page=2]{bipartite_outerplanar.pdf}
    \subcaption{construction for a tree\label{SUBFIG:outerTree}}
    \end{minipage}
    \hfill
     \begin{minipage}[t]{.5\textwidth}
    \centering
    \includegraphics[page=1]{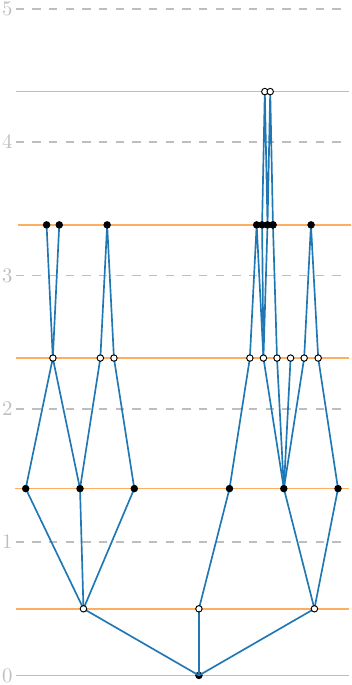}
    \subcaption{geometric realization when $G_s$ is the graph in~(c)\label{SUBFIG:outerOuter}}
    \end{minipage}
    \hfill\\[2\baselineskip]
    
    \begin{minipage}[t]{.6\textwidth}
    \centering
    \includegraphics[page=3]{bipartite_outerplanar}
    \subcaption{bipartite outerplanar graph\label{SUBFIG:outerGraph}}
    \end{minipage}
	\caption{\label{fig:tree}How to construct a geometric realization of a bipartite dichotomous ordinal graph if the short edges induce an outerplanar graph.}
\end{figure}

\begin{theorem}\label{thm:bipartiteOuter}
		A bipartite dichotomous ordinal graph admits a geometric realization if the subgraph induced by the short edges is outerplanar.
\end{theorem}
	\begin{proof}
	Let $G=(V,E_s \cup E_\ell)$ be a bipartite dichotomous ordinal graph such that $G_s=(V,E_s)$ is outerplanar.
	By \cref{obs:connected}, we may assume that $G_s$ is connected. We root $G_s$ at an arbitrary vertex $r$. Let $V_k$, $k=0,\dots$ be the BFS layers of $G_s$ rooted at $r$, i.e., $V_0 = \{r\}$, $V_1$ is the set of neighbors of $r$, and $V_{k+1}$, $k \geq 1$ is the set of neighbors of the vertices in $V_{k}$ that are not already in $V_{k-1}$. 	We say the $w$ is a \emph{child} of $v$ and $v$ is a \emph{parent} of $w$ if $vw$ is an edge of $G_s$, $v \in V_k$ and $w \in V_{k+1}$ for some $k$. Observe that by outerplanarity each vertex has at most two parents. We  construct a planar drawing of $G_s$ with the following properties.
	\begin{itemize}
		\item 
		The root $r$ is drawn with y-coordinate $y_0= 0$.    
		All vertices in layer $V_k$, $k>0$ are on a horizontal line $\ell_k$ with y-coordinate $y_k$ strictly between $k-1$ and $k$.
		\item
		The distance between a vertex and its children is at most 1 while the distance between two vertices on consecutive layers is greater than 1 if they are not adjacent in $G_s$.
		\item
		For each vertex $v$ there is a vertical strip $S_v$ such that
        \begin{itemize}
			\item $v$ is in $S_v$,
			\item $S_w$ is contained in the union of the strips of $w$'s parents,
			\item $S_u$ and $S_v$ are internally disjoint if $u$ and $v$ are  on the same layer.
		\end{itemize}
	\end{itemize}
	We construct the drawing as follows (see \cref{SUBFIG:outerTree}):
	Let $y_1=1/2$. Let $C$ be the circle with radius 1 centered at $r$. Let $x^\ell$ and $x^r$ be the x-coordinates of the intersection points of $C$ with $\ell_1$. Distribute  $V_1$ on $\ell_1$ 
	between $x^\ell$ and $x^r$. 
	$S_r$ is the whole plane.
	
	Assume now that we have drawn the vertices of layers $V_0,\dots,V_{k}$  and that we have established vertical strips for vertices in  $V_0,\dots,V_{k-1}$ for some $k$. 
	We now define the y-coordinate $y_{k+1}$ of the horizontal line $\ell_{k+1}$, place the vertices of layer $V_{k+1}$, and establish the vertical strips for the vertices on layer $k$. 
	
	
	For each vertex $v$ in layer $k$, consider the circle $C_v$ with radius 1 centered at $v$. 
	Let $y$ be the y-coordinate of the topmost intersection point between any two cycles $C_v$ and $C_w$ with $v,w \in V_k$. If $V_k$ contains only one element or if no two circles $C_u$ and $C_v$, $u,v \in V_k$ intersect then set $y=k$.  Choose  $y_{k+1}$ such that $y \leq y_{k+1} < y_k + 1$ and $y_{k+1} > k$. 
	For a vertex $v \in V_k$ with parent $u$ let $x_v^\ell$ and $x_v^r$, respectively, be the leftmost and rightmost intersection point of~$C_v$ with~$\ell_{k+1}$. If $x_v^\ell$ is to the left of $S_u$ or $x_v^r$ is to the right of $S_u$, respectively, then replace $x_v^\ell$ or $x_v^r$ by the respective boundary of $S_u$.
	We call the part of $\ell_{k+1}$ strictly between $x_v^\ell$ and $x_v^r$ the \emph{range} $\ell_{k+1}(v)$ of $v$.
	Distribute the children of $v$ on $\ell_{k+1}(v)$, maintaining the embedding.

	Special care has to be taken if a vertex $w \in V_{k+1}$ has two parents $u$ and $v$, i.e., if $w$ closes an internal face. In that case, we have to make sure that on one hand the intersection point $p$ of $C_u$ and $C_v$ is contained in $\ell_{k+1}$. On the other hand, we require that $p$ is on the boundary of $S_u$ and~$S_v$. Then we can place $w$ on $p$. The two conditions are true if
	\begin{enumerate}
		\item \label{item:close}
		$u$ and $v$ are \emph{close}, i.e., if the distance between $u$ and $v$ is minimum between any two vertices on the same layer $V_k$ and additionally small enough such that the intersection point $p$ of $C_u$ and $C_v$ has y-coordinate greater than $k$ and
		\item \label{item:equi} $u$ and $v$ are equidistant from the common boundary of their strips $S_u$ and~$S_v$.
	\end{enumerate}

	In the following, we discuss how \cref{item:close,item:equi} can be achieved.
	Consider an internal face~$f$. Let $V_j,\dots,V_k$ be the BFS-layers that contain vertices incident to~$f$. Observe that layers $V_{j+1},\dots,V_{k-1}$ contain pairs of vertices incident to $f$, each of which is consecutive on the respective line. The two vertices contained in $V_{j+1},\dots,V_{k-1}$, respectively, are called a \emph{pair} of $f$. The pair in $V_{j+1}$ is the \emph{first pair} of $f$ and the pair in $V_{k-1}$ is the \emph{last pair} of $f$.  The goal is now to construct a drawing in which not only the last pair of each face is close and equidistant to the boundary of their strips, but actually all pairs of a face. 
	We proceed as follows. Line $\ell_0$ contains only the root $r$ and, thus, no pairs. So assume we have constructed the correct drawing up to line $\ell_k$. If the y-coordinate $y$ of the topmost intersection point among any two circles $C_u$ and $C_v$ for any two $u,v \in V_k$ is greater than $k$ then let   $\ell_{k+1}$ be the line with y-coordinate exactly $y$. 
	
	
	If we place a vertex $w$ with two neighbors $u,v \in V_k$ then, we draw $w$ on the topmost intersection point $t$ of $C_u$ and $C_v$. Recall that $t$ is on $\ell_{k+1}$ since $u$ and $v$ were close by construction. In a first step, the remaining vertices are distributed arbitrarily within the range of their predecessors. 
	
	If $V_{k+1}$ does not contain any pairs of any faces, we are done. Otherwise determine a distance $d>0$ that is (i) lower or equal the distance between any two vertices on line $\ell_{k+1}$,  (ii) lower or equal twice the distance between any vertex of a pair that is not a first pair and the strip boundary between the pair, and (iii)  small enough so that the respective circles intersect above the line with y-coordinate~$k+1$. 
	
	Now, we first process pairs of vertices that are not first pairs. Observe that each vertex can only be involved in one such pair. Otherwise the predecessor is not on the outer face. Let $u$ and $v$ be such a pair. Observe that neither $u$ nor $v$ can close another face. Again otherwise the predecessor is not on the outer face.
	If $u$ and $v$ had a distance greater than $d$ in the first place than we move $u$ and $v$ closer together so that they now have distance $d$ and are equidistant from the strip boundary between them. Observe that this is possible since the predecessors of $u$ and $v$ were already a pair and thus, by construction, close. Thus their respective circles intersect on $\ell_{k+1}$. Observe that with this operation, the minimum distance between any two vertices is still $d$. 
	
	Next, we process first pairs. Observe that no vertex can be involved in several first pairs. Otherwise it is not on the outer face. However, a vertex can be contained in a first and another pair and it is also possible that a vertex closing another face is contained in a first pair.  So some of the vertices might already have a fixed position. However, it is not possible that both vertices of a first pair are contained in another pair or close a face. Otherwise, the predecessor of the first pair was an internal vertex. So, we first process the first pairs that share a vertex with another pair or where one of its vertices close a face and reduce their distance to $d$. Then we can also process all the other first pairs. In the end all pairs have distance $d$ and $d$ is still the minimum distance between any two vertices on $\ell_{k+1}$. This concludes the construction for \cref{item:close,item:equi}.

	The strips are defined as follows: 
	Let $u \in V_{k-1}$ and let $v_1,\dots,v_h$ be the children of $u$ from left to right. We partition the strip $S_u$ into the strips $S_{v_1},\dots,S_{v_h}$ by splitting $S_u$ at an arbitrary spot between $x_{v_{i-1}}^r$ and $x_{v_i}^\ell$, for $i=2,\dots,h$. If a vertex has two parents then its strip is the union of the two portions obtained from its parents.
	
	By construction, all edges of $G_s$ have length at most one. 	It remains to show that edges in~$E_\ell$ have length greater than one. So let $w_jw_k \in E_\ell$ with $w_j \in V_j$, $w_k \in V_k$, and $j \leq k$. Since $G$ is bipartite, we can exclude $k=j$ and $k=j+2$. If $k-j \geq 3$ then $y_k-y_j> k-1-j \geq 2$. Hence, the distance between $w_j$ and $w_k$ is greater than one. 
	If $k-j=1$ then the parent  $w$ if $w_k$ is in $V_j$. By construction, we know that the distance between $w_j$ and any child of $w$ is greater than one.
\end{proof}

	\begin{figure}[bhtp]
		\centering
		\includegraphics[page=1]{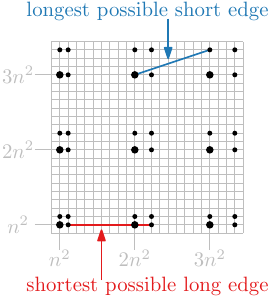}
		\caption{
  For each grid point $(i,j)$, $1\leq i \leq n$, $1\leq j \leq n$ there are four possible points. 
  If $i>1$, the x-coordinate is $in^2$ if the edge between $(i-1,j)$ and $(i,j)$ is short and $in^2+i$ otherwise. 
  If $j>1$, the y-coordinate is $jn^2$ if the edge between $(i,j-1)$ and $(i,j)$ is short and $jn^2+j$ otherwise. }
		\label{fig:grid}
	\end{figure}
 \begin{theorem}
 \label{thm:grid}
        A dichotomous ordinal graph $G=(V,E_s \cup E_\ell)$ admits a geometric realization if the set of short edges induces a subgraph of the grid.
 \end{theorem}
	\begin{proof}
	Extend $G_s=(V,E_s)$ by the remaining grid edges and require the new edges to be long. 
	We now perturb the grid.
	For each grid point $(i,j)$, $1\leq i \leq n$, $1\leq j \leq n$ of the original grid,  there are  four possible choices with the x-coordinates $in^2$  or $in^2+i$ and the y-coordinates $jn^2$  or $jn^2+j$. See Fig.~\ref{fig:grid}.
	In the case where  $i=1$ 
	choose as x-coordinate $n^2$ or $n^2+1$. Similarly, if 
	$j=1$  choose as y-coordinate  $n^2$ or $n^2+1$.
	If $i>1$, we choose the x-coordinate $in^2$ if the edge between $(i-1,j)$ and $(i,j)$ is short and $in^2+i$ otherwise. Analogously, if $j>1$, we choose the y-coordinate $jn^2$ if the edge between $(i,j-1)$ and $(i,j)$ is short and $jn^2+j$ otherwise. 
	
	Each long edge has length at least $n^2+1$ while each short edge has length at most $\sqrt{(n^2)^2 + n^2} < \sqrt{(n^2)^2 + n^2 + \frac14} = n^2 + \frac{1}{2}$.  By scaling the drawing, we obtain  a geometric realization of the  dichotomous ordinal graph.
\end{proof}
	


\begin{theorem}
\label{thm:long}
  A dichotomous ordinal graph admits a geometric realization in~$\R^2$ and on~$\Sp^2$ if each component of the long subgraph forms a caterpillar.
\end{theorem}
\begin{proof}\label{proof:thmlong}
  For our construction we may assume without loss of generality (cf.~\cref{obs:delta}) that the long edges are those of length at least one. Let~$G$ be a caterpillar and let~$v_1,\ldots,v_k$ be the spine path of~$G$. We draw the vertices on a circle~$C$ of radius~$\frac12+\varepsilon$ centered at the origin~$O$, for some suitably small~$\varepsilon>0$. For a point~$p\in C$ there are two points~$\mathrm{ucw}(p),\mathrm{uccw}(p)\in C$ that are at unit distance to~$p$ such that~$\measuredangle pO\mathrm{ucw}(p)=\pi+\varepsilon'$ and~$\measuredangle pO\mathrm{uccw}(p)=\pi-\varepsilon'$, for some~$\varepsilon'$ that depends on~$\varepsilon$. The points on~$C$ that are at distance at least one to~$p$ lie on the circular arc~$\mathrm{uccw}(p)\mathrm{ucw}(p)$, which spans an angle of~$2\varepsilon'$.
  
  We place~$v_1$ arbitrarily on~$C$ and then~$v_i$, for~$i=2,\ldots,k$, at~$\mathrm{ucw}(v_{i-1})$. Choosing~$\varepsilon$ so that~$\varepsilon'<\pi/(2k)$ ensures that we obtain a valid representation of the spine path of~$G$, that is, the long edges spanned by the points are exactly the edges of the spine path. It remains to place leaves that may be attached to the spine vertices. Let~$\alpha=\varepsilon'/(2k)$. We place all (non-spine) leaves adjacent to~$v_i$ in~$G$, for~$i\in\{1,\ldots,k\}$, close to the point~$v_i'$ that is obtained by shifting the point antipodal to~$v_i$ in~$C$ counterclockwise by an angle of~$i\alpha$ along~$C$. More precisely, we place these leaves uniformly spaced within an arc~$A_i$ of~$C$ with central angle~$\alpha/4$ whose midpoint is~$v_i'$. All points on~$A_i$ are far from~$v_i$ but close to all points in~$A_{i-1}$. Thus, we obtain a valid representation of~$G$ on~$C\simeq\Sp^2$.  
\end{proof}

\section{Pandichotomic Dimension and Degeneracy}\label{sec:3gen}

In this section, we exhibit a strong connection between the degeneracy of a graph and its pandichotomic dimension. Among others, we show 
that the pandichotomic dimension is bounded both from above and from below by a linear function of the degeneracy, and we give bounds on the coefficients in these linear functions.

%


\begin{theorem}\label{thm:3genr3}
  Every $d$-degenerate graph, for~$d\ge 2$, is pandichotomous on~$\Sp^{d-1}$ and in~$\R^d$.
\end{theorem}
\begin{proof}
  Let~$G$ be a $d$-degenerate dichotomous ordinal graph, for~$d\ge 2$, and let~$u_1,\ldots,u_n$ be a vertex ordering such that~$u_i$ has at most~$d$ neighbors in~$\{u_1,\ldots,u_{i-1}\}$, for~$1\le i\le n$. To obtain a geometric realization of~$G$ we place the vertices on the sphere~$S=\{p\in\R^d\colon\lVert p\rVert=\sqrt{2}/2\}$ such that for every $d$-tuple of vertices the corresponding vectors are linearly independent. For a vertex~$p\in S$ denote by~$\vec{v}_p$ the corresponding vector; 
  the points in distance less than one to~$p$ on~$S$ form a hemisphere, which consists of all points~$q\in S$ such that~$\vec{v}_p\vec{v}_q>0$, where $\vec{v}_p\vec{v}_q$ denotes the scalar product between $\vec{v}_p$ and $\vec{v}_q$. 
   
  We place the vertices~$u_1,\ldots,u_n$ in this order. The first vertex~$u_1$ is placed arbitrarily on~$S$. Then for each~$u_i$, for~$2\le i\le n$, we have 
  a set~$Q=\{q_1,\ldots,q_t\}\subset S$ of points, which correspond to the~$t\le d$ neighbors of~$u_i$ in~$u_1,\ldots,u_{i-1}$.
  Further, for each point in~$Q$ we know whether it should be close to or far from~$u_i$. Now we need to find a suitable point on~$S$ that satisfies these constraints. We obtain the set~$Q'=\{q'_1,\ldots,q'_t\}$ from~$Q$ by setting~$q'_j=q_j$, if~$q_j$ should be close to~$u_i$ and setting~$q'_j$ to the point antipodal to~$q_j$ on~$S$ if~$q_j$ should be far from~$u_i$. As~$Q$ is linearly independent, so is~$Q'$. Furthermore, the linear system~$\vec{v}_{q_j'}\vec{x}=\vec{1}$, for~$1\le j\le t$, has a (unique) solution~$\vec{x}_0$. Normalizing~$\vec{x}_0$ we obtain a point~$r\in S$, such that~$\vec{v}_r\vec{v}_{q_j'}>0$, for all~$1\le j\le t$. Thus, we can place~$u_i$ at~$r$ such that all constraints with respect to~$Q$ are satisfied. The points on~$S$ that lie in a $k$-dimensional subspace spanned by~$k$ vertices, for~$k<d$, can easily be avoided when selecting~$r$. In this way we ensure that for all $d$-tuples of vertices on~$S$ the corresponding vectors are linearly independent. 
  This realization also works for geodesic distances on~$\Sp^{d-1}$ if we take an arc with angle~$\pi/2$ as a unit. 
\end{proof}


\begin{corollary}\label{cor:planar-on-sphere}
  Every bipartite planar graph is pandichotomous on~$\Sp^2$ and in~$\R^3$. 
\end{corollary}
\begin{proof}
  Let~$G=(V,E)$ be a planar bipartite graph on~$n$ vertices. As a consequence of Euler's formula, we have~$|E|\le 2n-4$. By the handshaking lemma~$\sum_{v\in V}\deg(v)=2|E|\le 4n-8$ and, thus, the average degree in~$G$ is strictly less than~$4$. Consequently, every planar bipartite graph is~$3$-degenerate, and the statement follows by \cref{thm:3genr3}.
\end{proof}

We contrast the upper bound on the pandichotomic dimension in terms of the degeneracy provided by \cref{thm:3genr3} by an almost matching lower bound in the Euclidean (\cref{thm:3genr3-ub}) and a matching lower bound in the spherical case (\cref{thm:3genr3-ub-s}), even for bipartite graphs. Note that every lower bound example is not  just sporadic but spans an infinite family of examples, given that being pandichotomous is a monotone graph property.

\begin{theorem}\label{thm:3genr3-ub}
  For every~$d\ge 2$, there exists a $(d+2)$-degenerate bipartite 
  graph that is not 
  pandichotomous in~$\R^d$.
\end{theorem}
\begin{proof}
  We build a graph~$G=(V,E)$ starting with a set~$V_0$ of~$d+2$ isolated vertices. 
  Then for each of the~$2^{d+2}$ many choices of assigning short or long to the vertices of~$V_0$ we add a new vertex and connect it by edges to~$V_0$ accordingly. Observe that~$G$ is bipartite and $(d+2)$-degenerate by construction. Consider a geometric realization of~$G$ and the corresponding arrangement~$\mathcal{A}$ of the~$d+2$ unit spheres centered at the vertices of~$V_0$.
  In order to bound the number of full-dimensional cells in~$\mathcal{A}$, we consider the standard parabolic lifting map~$\lambda$, which orthogonally projects a point in~$\R^d$ to the unit paraboloid in~$\R^{d+1}$. This map has a number of useful properties, see, for instance, Edelsbrunner's book~\cite[Section~1.4]{e-acg-87}. Specifically, the image~$\lambda(S)$ of a sphere~$S\subset\R^d$ lies on a hyperplane~$h_S\subset\R^{d+1}$ such that a point~$p\in\R^d$ lies inside~$S$ if and only if its image~$\lambda(p)$ lies below~$h_S$. It follows that we may regard~$\mathcal{A}$ as an arrangement of~$d+2$ hyperplanes in~$\R^{d+1}$. As such an arrangement has at most~$\sum_{i=0}^{d+1}\binom{d+2}{i}=2^{d+2}-1<2^{d+2}$ many full-dimensional cells~\cite[Proposition~6.1.1]{m-ldg-02},
  for at least one of the vertices in~$V\setminus V_0$ its distance requirements with respect to~$V_0$ are violated. Thus, there is no geometric realization of~$G$ in~$\R^d$.
\end{proof}
\begin{corollary}\label{thm:3genr3-ub-s}
  For every~$d\ge 2$, there exists a $(d+1)$-degenerate bipartite dichotomous ordinal graph that is not realizable on~$\Sp^{d-1}$.
\end{corollary}

\begin{proof}
  We use the same graph~$G$ as in the proof of \cref{thm:3genr3-ub}, except that we start with a set~$V_0$ of~$d+1$ rather than~$d+2$ isolated vertices. We consider~$\Sp^{d-1}$ as a unit sphere~$S\subset\R^d$. For a point~$p\in S$, the set of points on~$S$ in distance at most some fixed unit from~$p$ can be expressed as an intersection~$S\cap u_p$ where~$u_p$ is a halfspace. 
  
  Suppose there is some geometric realization of~$G$, and let~$P\subset S$ denote the set of~$d+1$ points that represent the vertices of~$V_0$ in this realization. Consider the set~$H$ of the~$d+1$ bounding hyperplanes of the halfspaces~$u_p$, for~$p\in P$. The arrangement~$\mathcal{A}$ of~$H$ has at most~$\sum_{i=0}^{d}\binom{d+1}{i}=2^{d+1}-1<2^{d+1}$ many full-dimensional cells~\cite[Proposition~6.1.1]{m-ldg-02}. Thus, for at least one of the vertices in~$V\setminus V_0$ there is no cell that satisfies its distance requirements with respect to~$V_0$. It follows that there is no geometric realization of~$G$ in~$\Sp^{d-1}$.
\end{proof}

In the Euclidean case, we give a tight lower bound in~$\R^2$, albeit not for bipartite graphs.

\begin{figure}[htbp]
    \centering
    \includegraphics[page=1]{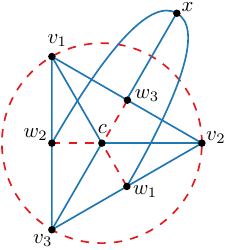} \hfil
    \includegraphics[page=2]{3-degen}    
    \caption{A 3-degenerate graph that cannot be realized in the plane along with a sketch of the proof. Long edges are dashed.  The three points $v_1,v_2,v_3$ would have to lie in the intersection of the disk $B(c,1)$ and a unit disk 
    within $B(w_1,1) \cup B(w_2,1) \cup B(w_3,1)$. But there is not enough space for three points with pairwise distance at least one.}
    \label{fig:3-degen-plane}
\end{figure}


\begin{theorem}\label{thm:3deg2}
  There exists a~$3$-degenerate 
  graph that is not 
  pandichotomous in~$\R^2$.
\end{theorem}

\begin{proof}
  We build a dichotomous ordinal graph~$G$ on eight vertices as follows 
  (\Cref{fig:3-degen-plane}). Start with a set~$V=\{v_1,v_2,v_3\}$ and a complete graph on~$V$, all 
  edges 
  long. Then insert a vertex~$c$ and connect it to all vertices in~$V$ by a short edge. Next, insert a set~$W=\{w_1,w_2,w_3\}$ of vertices as follows: the vertex~$w_i$, for~$i\in\{1,2,3\}$, is connected to~$c$ by a long edge and to each vertex in~$V_i^-:=V\setminus\{v_i\}$ by a short edge. Finally, we add a vertex~$x$ that is connected to all vertices in~$W$ by a short edge. Observe that~$G$ is $3$-degenerate as constructed. 

  Consider any geometric realization~$\Gamma$ of~$G$. In a slight abuse of notation we identify the vertices of~$G$ with the corresponding points that represent them in~$\Gamma$. For a set~$P\subset\R^2$ denote~$I(P)=\bigcap_{p\in P}\mathrm{B}(p,1)$ and~$U(P)=\bigcup_{p\in P}\mathrm{B}(p,1)$. As~$x\in I(W)$, we have~$I(W)\ne\emptyset$ and~$U(W)$ is simply connected. The constraints imposed by~$G$ enforce~$c\notin U(W)$ and~$V\subset\mathrm{B}(c,1)\cap U(W)$. In particular, for each~$i\in\{1,2,3\}$, the set~$V_i^-$ is contained in the lens~$L_i=I(\{c,w_i\})$. As~$c\notin U(W)$, each~$L_i$ has diameter strictly smaller than~$\sqrt{3}$. Any lens in~$\R^2$ of diameter strictly smaller than~$\sqrt{3}$ contains at most two points at a pairwise distance at least one. Thus, in particular, we have~$v_i\notin L_i$, for all~$i\in\{1,2,3\}$. 

  We analyze the interaction of the circles~$S_i=\mathrm{S}(w_i,1)$, for~$i\in\{1,2,3\}$, with~$C=\mathrm{S}(c,1)$. As~$|L_i\cap V|\ge 2$ and as any two points in~$V$ are at distance strictly larger than one, we have~$|S_i\cap C|=2$, for all~$i\in\{1,2,3\}$. As~$c\notin U(W)$ and as~$U(W)$ is simply connected, some part of~$C$ is disjoint from~$U(W)$. We use such a part to break up the circular sequence of intersection points between~$C$ and the~$S_i$ into a linear sequence~$\mathcal{I}$ of six points (some but not all of which may coincide). On the one hand, as~$|L_i\cap V|\ge 2$, for all~$i\in\{1,2,3\}$, and as~$|V|=3$, we have~$L_i\cap L_j\ne\emptyset$, for all~$i,j\in\{1,2,3\}$. On the other hand, if~$L_i\subseteq L_j$, for some~$i,j\in\{1,2,3\}$ with~$i\ne j$, then, given that~$v_j\in V_i$ we also have~$v_j\in L_j$, a contradiction. It follows that up to the indexing of~$W$ we have~$\mathcal{I}=1,2,3,1,2,3$, where we write an integer~$i$ to indicate an intersection~$C\cap S_i$. 

  Let~$q$ denote the intersection point of~$S_1\cap S_3$ in~$\mathrm{B}(c,1)$. Note that~$q\notin L_2$ because~$q\in L_2$ would imply~$L_1\cap L_3\subseteq L_2$ and thus~$v_2\in L_2$, a contradiction. We continuously move a unit disk~$D$ starting from~$D=\mathrm{B}(w_2,1)$ towards~$c$ such that the center~$d$ of~$D$ is on the line segment~$w_2c$. We stop the movement as soon as~$d\in\mathrm{S}(q,1)$. This process is well defined because~$q\notin\mathrm{B}(w_2,1)$ and~$q\in L_1\cap L_3\subset\mathrm{B}(c,1)$. Furthermore, during the whole movement we have~$L_2\subset D\subset U(W)$ and the disk~$D'$ at the end of the movement also contains~$L_1\cap L_3$. Thus, the lens~$L=D'\cap\mathrm{B}(c,1)$ contains all of~$V$ and given that~$L\subset U(W)$ it has diameter strictly smaller than~$\sqrt{3}$, a contradiction.
\end{proof}

The examples discussed above demonstrate that the bound from \cref{thm:3genr3} is tight in the worst case, that is, for \emph{some} graphs. In the following, we will show that the bound is also tight for \emph{all} graphs, up to a multiplicative constant. 
We first estimate the edge density of pandichotomous graphs in $\R^d$ and show that it is linearly bounded in $d$. In fact, we prove a stronger result, namely that every sufficiently dense graph $G$ not only is not pandichotomous in $\R^d$, but in fact, asymptotically \emph{almost all} partitions of the edge-set of $G$ into short and long edges yield dichotomous ordinal graphs that  are not realizable in $\R^d$. 

\begin{theorem}\label{thm:universalindspace}
    Let $d\in\N$, let~$\varepsilon>0$, and let $G=(V,E)$ be a graph with~$n$ vertices and~$m$~edges. Let~$c<7.182$ denote the unique positive root of the function~$x-3-x\cdot H(1/x)$, where~$H\colon(0,1)\to\R$ is the binary entropy function $H(x)=-x\log_2(x)-(1-x)\log_2(1-x)$. 
    
    If~$m> (c+\varepsilon)dn$, then for asymptotically almost all\/\footnote{That is, a fraction of the partitions that tends to $1$ as $dn\rightarrow \infty$.} partitions~$E=E_s\cup E_\ell$ of the edge-set of~$G$, the dichotomous ordinal graph~$(V,E_s\cup E_\ell)$ has no geometric realization in~$\R^d$. In particular, every graph that is pandichotomous in $\R^d$ has at most $(c+o(1))dn$ edges.
    The~$o(1)$ term approaches zero quickly as~$dn\to\infty$. For illustration, we also show that~$m<\mu dn$, for all~$n\in\N$ and~$d\ge 2$ and some explicit constant $\mu<7.23$. 
\end{theorem}
In the proof of \cref{thm:universalindspace}, we will use the following 
result by Warren~\cite{warren:68} on the number of sign-patterns of a collection of multivariate polynomials. 

\begin{theorem}[{Warren~\cite[Theorem~2]{warren:68}}]\label{thm:warren}
Let $q_1,\ldots,q_m$ be real polynomials in $N$ variables, each of degree at most~$d$. Then the number of connected components of the set 
\[
\left\{\left.\mathbf{x}\in \R^N\:\right|\:\forall i \in \{1,\ldots,m\}: q_i(\mathbf{x})\neq 0\right\}
\]
is at most
\[
2(2d)^N\sum_{k=0}^{N}{2^k\binom{m}{k}}.
\]
\end{theorem}
The following is an immediate consequence of the above.

\begin{corollary}\label{cor:signpatterns}
 Let $q_1,\ldots,q_m$ be real polynomials in $N$ variables, each of degree at most~$d$ and let 
 $\mathcal{S}:=\{(\mathrm{sgn}(q_1(\mathbf{x})),\ldots,\mathrm{sgn}(q_m(\mathbf{x}))) \mathop{\big\vert} \forall i \in \{1,\ldots,m\}: q_i(\mathbf{x})\neq 0\} \subseteq \{+1,-1\}^m$. Then we have
 \[
 |\mathcal{S}|\le 2(2d)^N\sum_{k=0}^{N}{2^k\binom{m}{k}}.
 \]
\end{corollary}
\begin{proof}
Let $X:=\{\mathbf{x}\in \R^N\mathop{\big|}\forall i \in \{1,\ldots,m\}: q_i(\mathbf{x})\neq 0\}$. Consider the map
$f:X\rightarrow \{+1,-1\}^m$, defined by 
$f(\mathbf{x})=(\mathrm{sgn}(q_1(\mathbf{x})),\ldots,\mathrm{sgn}(q_m(\mathbf{x})))$.
It is easy to see, using the continuity of the polynomials $q_1,\ldots,q_m$, that $f$ is 
continuous. 
Let $C$ be a topological component of $X$. The restriction of $f$ to $C$ is a continuous map on a connected topological space that takes on only finitely many (at most $2^m$) values. Since every such map is constant, we find that $f$ is constant on every component of $X$. Hence, the size of $\mathcal{S}=\mathrm{im}(f)$ is bounded by the number of components of $X$. The statement now follows from \cref{thm:warren}.
\end{proof}

\begin{proof}[Proof of \cref{thm:universalindspace}]
Assume w.l.o.g.~that $G=(V,E)$ has vertex set $V=\{1,\ldots,n\}$.
    Let $(\mathbf{x}_1,\ldots,\mathbf{x}_n)$ be an $n$-tuple of points in $\R^d$, no two of which are at distance exactly one. The \emph{distance-pattern} of $(\mathbf{x}_1,\ldots,\mathbf{x}_n)$ w.r.t.~$G$ is the map $s_G(\mathbf{x}_1,\ldots,\mathbf{x}_n):E\rightarrow \{+1,-1\}$, where 
    \[
    s_G(\mathbf{x}_1,\ldots,\mathbf{x}_n)(ij):=\begin{cases}
        +1, & \text {if }\dist{\mathbf{x}_i}{\mathbf{x}_j}>1, \cr -1, & \text{ if }\dist{\mathbf{x}_i}{\mathbf{x}_j}<1, 
    \end{cases}
    \]
    for every edge $ij\in E$. In other words, given a configuration of points $(\mathbf{x}_1,\ldots,\mathbf{x}_n)$ in $\R^d$, the distance-pattern captures the information of which edges of $G$ in a geometric realization placing vertices $1,\ldots,n$ at points $\mathbf{x}_1,\ldots,\mathbf{x}_n$, respectively, are long or short. 

    Next, we estimate the size of the set $\mathcal{S}$ of all possible distance-patterns for $G$. Concretely, $\mathcal{S}$ is the set of all mappings $s_G(\mathbf{x}_1,\ldots,\mathbf{x}_n)$ where $(\mathbf{x}_1,\ldots,\mathbf{x}_n)$ ranges over all $n$-tuples of points in $\R^d$, no two of which are at distance $1$. 
    %
    Let us denote by $M$ the number of partitions $E_s\cup E_\ell=E$ of the edge-set of $G$ such that the dichotomous ordinal graph $(V,E_s\cup E_\ell)$ is realizable in $\R^d$. Let $p$ denote the fraction of such partitions compared to all possible partitions of $E$. Note that $M=p2^m$ and that $p=1$ if and only if $G$ is pandichotomous. 
    \subparagraph{\textbf{Claim 1.}} $|\mathcal{S}|\ge M$. 
    \begin{proof}[Proof of Claim~1.]
    Consider any one of the $M$ partitions $E_s\cup E_\ell$ of the edge-set $E$ for which $(V,E_s\cup E_\ell)$ is realizable in $\R^d$. 
    By \cref{obs:delta} we may assume 
    a threshold of~$\delta=1$ and that we use a set of points with no unit distances. The distance-pattern of the corresponding tuple of $n$ points in $d$-space then has value $+1$ for every long and $-1$ for every short edge. 
    
    Since any two distinct partitions of $E$ into short and long edges result in different distance-patterns in $\mathcal{S}$ this way, it follows that $|\mathcal{S}|\ge M$, as claimed. 
    \end{proof}

    Next, we obtain an upper bound on $|\mathcal{S}|$ using \cref{thm:warren}. 

    \subparagraph{\textbf{Claim 2.}} 
    $
    \displaystyle{|\mathcal{S}|\le 2\cdot 4^{dn}\cdot \sum_{k=0}^{dn}{2^k\binom{m}{k}}}.
    $
    \begin{proof}[Proof of Claim~2.]
    For any two points $\mathbf{x}, \mathbf{y} \in \R^d$, we have $\dist{\mathbf{x}}{\mathbf{y}}>1$ if and only if $\dist{\mathbf{x}}{\mathbf{y}}^2-1>0$, and $\dist{\mathbf{x}}{\mathbf{y}}<1$ if and only if $\dist{\mathbf{x}}{\mathbf{y}}^2-1<0$. 
    Further, the expression 
    \[
    \dist{\mathbf{x}}{\mathbf{y}}^2-1=(x_1-y_1)^2+(x_2-y_2)^2+\ldots+(x_d-y_d)^2-1
    \]
    is a polynomial of degree two in the coordinates of $\mathbf{x}, \mathbf{y}$.
    It is therefore natural to associate with every edge $ij \in E$ of $G$ a polynomial $q_{i,j}$, defined as
    $q_{i,j}(\mathbf{x}_1,\ldots,\mathbf{x}_n):=\dist{\mathbf{x}_i}{\mathbf{x}_j}^2-1$. The collection $\left(q_{i,j}\mathop{\big|} ij\in E\right)$ then forms a set of $m$ polynomials, each of degree $2$, in the $N:=dn$ variables corresponding to the coordinates of $\mathbf{x}_1,\ldots,\mathbf{x}_n$. For every tuple $(\mathbf{x}_1,\ldots,\mathbf{x}_n)$ of points in $\R^d$, no two at unit distance, we 
    have 
    $s_G(\mathbf{x}_1,\ldots,\mathbf{x}_n)(ij)=\mathrm{sgn}(q_{i,j}(\mathbf{x}_1,\ldots,\mathbf{x}_n))$.

    Hence, the set $\mathcal{S}$ defined before coincides with the set $\mathcal{S}$ defined in \cref{cor:signpatterns} for the collection $\left(q_{i,j}\mathop{\big|} ij\in E\right)$ of $m$ multivariate polynomials. We therefore obtain, using the degree bound $2$ for the polynomials,
    \[
    |\mathcal{S}|\le 2\cdot (2\cdot 2)^{N}\cdot \sum_{k=0}^{N}{2^k\binom{m}{k}} =2\cdot 4^{dn}\cdot \sum_{k=0}^{dn}{2^k\binom{m}{k}}.\qedhere
    \]
    \end{proof}

    Using Claims~1 and~2, we obtain the inequality
    \[
    p2^m=M\le |\mathcal{S}|\le 2\cdot 4^{dn}\cdot \sum_{k=0}^{dn}{2^k\binom{m}{k}}.
    \]

Since $c>2$, we have $m\ge (c+\varepsilon)dn>2dn$. This 
implies that the term $\binom{m}{k}$ is monotonically increasing for $k=0,1,\ldots,dn$. Hence, we may bound the right hand side above by
\[
2\cdot 4^{dn}\cdot\binom{m}{dn}\cdot\sum_{k=0}^{dn}2^k
= 
2\cdot 4^{dn}\cdot\binom{m}{dn}\cdot(2^{dn+1}-1)
<
4\cdot 8^{dn}\cdot\binom{m}{dn}.
\]
Let $x:=\frac{m}{dn}\ge c+\varepsilon$. Using the estimate~$\binom{a}{\lambda a}\le 2^{aH(\lambda)}$, which holds for any $a\in\N$ and $\lambda \in (0,1)$, see e.g., \cite[Lemma~9.2]{mu-pcrapa-05}, we find
\[
p2^{xdn}=p2^m
<4\cdot 
8^{dn}2^{xdn\cdot H(1/x)}.
\]
Taking logarithms 
and dividing by $dn$ this simplifies to
\begin{equation}\label{eq:H}
x<\frac{2
-\log_2(p)}{dn}+3+x\cdot H(1/x).
\end{equation}

Consider the function~$f(x)=x-3-x\cdot H(1/x) = x -3 -x \log_2x + (x-1)\log_2(x-1)$. Observe that~$f$ is strictly monotonically increasing and continuous for~$x\ge 2$ and that it has a unique positive root~$c\approx 7.1815$. 

It follows that $f(c+\varepsilon)>f(c)=0$, and, since $x\ge c+\varepsilon$, \cref{eq:H} above implies that 
\[
f(c+\varepsilon)\le f(x)\le \frac{2
-\log_2(p)}{dn} 
\Longrightarrow
\log_2(p)\le 2
-f(c+\varepsilon)dn.
\]
Since the right hand side of the last inequality tends to $-\infty$ as $dn\rightarrow \infty$, it follows that $p\rightarrow 0$ for $dn\rightarrow \infty$. Hence, it is indeed true that asymptotically almost all partitions $E_s\cup E_\ell=E$ of the edge-set of $G$ into short and long edges yield dichotomous ordinal graphs with no geometric representation in $\R^d$. This completes the proof of the first statement.

It remains to show that if $G$ is pandichotomous ($p=1$), then the explicit bound~$m\le \mu dn$ with $\mu:=7.2240208$
holds for all~$n\in\N$ and~$d\ge 2$. Set~$\phi(z)=2/z$ 
and observe that~$\phi\to 0$ for~$z\to\infty$. 
We can then express \cref{eq:H} as $f(x)<\phi(dn)$.
 The trivial bound~$m\le\binom{n}{2}$ suffices to show that $m\le \mu dn$, unless~$\mu dn<\binom{n}{2}$, i.e., unless $\lceil 4\mu+1\rceil\leq n$.
As~$\phi$ is monotonically decreasing, it follows that~$f(x)<\phi(dn)\le\phi(2n)\le\phi(2 \cdot \lceil 4\mu+1 \rceil)$. For the value~$\mu=7.2240208$
we obtain~$f(x)<\phi(60)
=1/30$, whereas~$f(\mu)>0.033333334>1/30$.
Therefore, as~$f$ is monotonically increasing, it follows that~$\frac{m}{dn}=x<\mu$, as claimed.
\end{proof}


\begin{corollary}\label{cor:cdbound}
For some absolute constant $\mu<7.23$ we have $\frac{\mathrm{d}(G)}{2\mu}\le \mathrm{ped}(G)\le \mathrm{d}(G)$ for every graph $G$ and $\frac{\mathrm{d}(G)}{2\mu}-1\le \mathrm{psd}(G)\le \mathrm{d}(G)-1$ for every graph $G$ with $\mathrm{d}(G)\ge 2$.
\end{corollary}
\begin{proof}
  By \cref{thm:universalindspace} every graph~$G$ that is pandichotomous in~$\R^d$ has at most~$\mu dn$ edges, for some absolute constant~$\mu<7.23$. Thus, the minimum degree of~$G$ is at most~$\lfloor 2\mu d\rfloor$ and so~$G$ is~$k$-degenerate, for~$k=\lfloor 2\mu d\rfloor$. This implies $\mathrm{d}(G)\le 2\mu\cdot \mathrm{ped}(G)$ and so $\frac{
\mathrm{d}(G)}{2\mu}\le \mathrm{ped}(G)\le \mathrm{d}(G)$. Since $\mathrm{psd}(G)\ge \mathrm{ped}(G)-1$ we obtain the lower bound for the spherical dimension. The upper bounds follow from \cref{thm:3genr3}. 
\end{proof}

\begin{corollary}\label{cor:3genr3-ub}
Every $n$-vertex graph is pandichotomous in $\mathcal{S}^{n-2}$ and $\R^{n-1}$. But there exists an absolute constant $\mu<7.23$ such that for every $d<\frac{n-1}{4\mu}$ there exist $n$-vertex bipartite graphs that are not pandichotomous in~$\R^d$ or~$\mathcal{S}^{d-1}$.
\end{corollary}
\begin{proof}
Since $n$-vertex graphs are $(n-1)$-degenerate, the first statement is a direct consequence of \cref{thm:3genr3}. 
The complete bipartite $n$-vertex graph $K_{\lfloor n/2\rfloor, \lceil n/2\rceil}$ has minimum degree $\lfloor n/2\rfloor$. Hence, \cref{cor:cdbound} implies that $K_{\lfloor n/2\rfloor, \lceil n/2\rceil}$ is not pandichotomous in $\R^d$ for any dimension $d$ such that $2\mu d<\lfloor n/2 \rfloor$, so in particular for $d<(n-1)/(4\mu)$. 
\end{proof}

\section{Conclusion}

We 
study pandichotomous graphs with an emphasis on bipartite graphs on one hand and the relationship between the degeneracy and the pandichotomous dimension on the other hand. Some interesting open questions remain, such as:
\begin{itemize}
    \item Is every planar graph, planar $3$-tree, or planar bipartite graph pandichotomous in $\R^2$?
    \item Are bipartite dichotomous ordinal graphs realizable in $\R^2$ if the short graph is a $2$-tree?
    \item Characterize the (complete) bipartite graphs that are pandichotomous in~$\R^k$, for~$k\ge 3$.
    \item Given $d \geq 3$, is there a bipartite $d$-degenerate graph that is not pandichotomous in $\R^{d-1}$? We believe that the answer to this question should be positive.
\end{itemize}
{

\bibliographystyle{plainurl} 
\bibliography{references}

\begin{thebibliography}{10}

\bibitem{alamCFKKPT:dam17}
Md.~Jawaherul Alam, Steven Chaplick, Gasper Fijavz, Michael Kaufmann, Stephen~G. Kobourov, Sergey Pupyrev, and Jackson Toeniskoetter.
\newblock Threshold-coloring and unit-cube contact representation of planar graphs.
\newblock {\em Discret. Appl. Math.}, 216:2--14, 2017.
\newblock \href {https://doi.org/10.1016/J.DAM.2015.09.003} {\path{doi:10.1016/J.DAM.2015.09.003}}.

\bibitem{alamKPT:FUN14}
Md.~Jawaherul Alam, Stephen~G. Kobourov, Sergey Pupyrev, and Jackson Toeniskoetter.
\newblock Happy edges: Threshold-coloring of regular lattices.
\newblock In {\em Proc. 7th Internat. Conf. Fun with Algorithms ({FUN} 2014)}, volume 8496 of {\em LNCS}, pages 28--39. Springer, 2014.
\newblock \href {https://doi.org/10.1007/978-3-319-07890-8\_3} {\path{doi:10.1007/978-3-319-07890-8\_3}}.

\bibitem{alam_etal:wg15}
Md.~Jawaherul Alam, Stephen~G. Kobourov, Sergey Pupyrev, and Jackson Toeniskoetter.
\newblock Weak unit disk and interval representation of graphs.
\newblock In {\em Proc. 41st Internat. Workshop Graph-Theoretic Concepts in Computer Science (WG~2015)}, volume 9224 of {\em LNCS}, pages 237--251. Springer, 2015.
\newblock \href {https://doi.org/10.1007/978-3-662-53174-7_17} {\path{doi:10.1007/978-3-662-53174-7_17}}.

\bibitem{AngeliniBGS19}
Patrizio Angelini, Michael~A. Bekos, Martin Gronemann, and Antonios Symvonis.
\newblock Geometric representations of dichotomous ordinal data.
\newblock In {\em Proc. 45th Internat. Workshop Graph-Theoretic Concepts in Computer Science (WG~2019)}, volume 11789 of {\em LNCS}, pages 205--217. Springer, 2019.
\newblock \href {https://doi.org/10.1007/978-3-030-30786-8_16} {\path{doi:10.1007/978-3-030-30786-8_16}}.

\bibitem{Bennett1960}
Joseph~F. Bennett and William~L. Hays.
\newblock Multidimensional unfolding: Determining the dimensionality of ranked preference data.
\newblock {\em Psychometrika}, 25(1):27--43, 1960.

\bibitem{DBLP:journals/scw/ChenPW17}
Jiehua Chen, Kirk Pruhs, and Gerhard~J. Woeginger.
\newblock The one-dimensional {Euclidean} domain: finitely many obstructions are not enough.
\newblock {\em Social Choice and Welfare}, 48(2):409--432, 2017.
\newblock \href {https://doi.org/10.1007/s00355-016-1011-y} {\path{doi:10.1007/s00355-016-1011-y}}.

\bibitem{DBLP:journals/jal/DoignonF94}
Jean{-}Paul Doignon and Jean{-}Claude Falmagne.
\newblock A polynomial time algorithm for unidimensional unfolding representations.
\newblock {\em J. Algorithms}, 16(2):218--233, 1994.
\newblock \href {https://doi.org/10.1006/jagm.1994.1010} {\path{doi:10.1006/jagm.1994.1010}}.

\bibitem{e-acg-87}
Herbert Edelsbrunner.
\newblock {\em Algorithms in combinatorial geometry}, volume~10 of {\em EATCS Monographs on Theoretical Computer Science}.
\newblock Springer, 1987.
\newblock \href {https://doi.org/10.1007/978-3-642-61568-9} {\path{doi:10.1007/978-3-642-61568-9}}.

\bibitem{DBLP:conf/ijcai/ElkindL15}
Edith Elkind and Martin Lackner.
\newblock Structure in dichotomous preferences.
\newblock In {\em Proc. 24th {IJCAI}}, pages 2019--2025. {AAAI} Press, 2015.
\newblock URL: \url{https://ijcai.org/Abstract/15/286}.

\bibitem{h-cdgacmc-82}
Timothy~F. Havel.
\newblock {\em The combinatorial distance geometry approach to the calculation of molecular conformation}.
\newblock Ph.{D}. thesis, University of California, Berkeley, CA, 1982.

\bibitem{helly:23}
Eduard Helly.
\newblock Über {M}engen konvexer {K}örper mit gemeinschaftlichen {P}unkten.
\newblock {\em Jahresbericht der Deutschen Mathematiker-Vereinigung}, 32:175--176, 1932.
\newblock URL: \url{http://eudml.org/doc/145659}.

\bibitem{km-sdprg-12}
Ross~J. Kang and Tobias M{\"u}ller.
\newblock Sphere and dot product representations of graphs.
\newblock {\em Discrete Comput. Geom.}, 47(3):548--568, 2012.
\newblock \href {https://doi.org/10.1007/s00454-012-9394-8} {\path{doi:10.1007/s00454-012-9394-8}}.

\bibitem{kruskal-msogfnh-64}
Joseph Kruskal.
\newblock Multidimensional scaling by optimizing goodness of fit to a nonmetric hypothesis.
\newblock {\em Psychometrika}, 29(1):1--27, 1964.

\bibitem{Kruskal1964}
Joseph Kruskal.
\newblock Nonmetric multidimensional scaling: A numerical method.
\newblock {\em Psychometrika}, 29(2):115--129, 1964.

\bibitem{m-sgs-84}
Hiroshi Maehara.
\newblock Space graphs and sphericity.
\newblock {\em Discrete Appl. Math.}, 7(1):55--64, 1984.
\newblock \href {https://doi.org/10.1016/0166-218X(84)90113-6} {\path{doi:10.1016/0166-218X(84)90113-6}}.

\bibitem{m-ldg-02}
Ji{\v r}{\' i} Matou{\v s}ek.
\newblock {\em Lectures on Discrete Geometry}.
\newblock Springer, New York, NY, 2002.
\newblock \href {https://doi.org/10.1007/978-1-4613-0039-7} {\path{doi:10.1007/978-1-4613-0039-7}}.

\bibitem{mu-pcrapa-05}
Michael Mitzenmacher and Eli Upfal.
\newblock {\em Probability and Computing: Randomized Algorithms and Probabilistic Analysis}.
\newblock Cambridge, New York, NY, 2005.
\newblock \href {https://doi.org/10.1017/CBO9780511813603} {\path{doi:10.1017/CBO9780511813603}}.

\bibitem{p-rmep-16}
Dominik Peters.
\newblock Recognising multidimensional euclidean preferences.
\newblock {\em CoRR}, abs/1602.08109, 2016.
\newblock URL: \url{10.48550/arXiv.1602.08109}.

\bibitem{peters:aaai17}
Dominik Peters.
\newblock Recognising multidimensional {Euclidean} preferences.
\newblock In {\em Proc. 21st AAAI Conf. Artificial Intelligence (AAAI'17)}, pages 642--648, 2017.
\newblock \href {https://doi.org/10.1609/AAAI.V31I1.10616} {\path{doi:10.1609/AAAI.V31I1.10616}}.

\bibitem{Shepard1962}
Roger~N. Shepard.
\newblock The analysis of proximities: Multidimensional scaling with an unknown distance function. {I}.
\newblock {\em Psychometrika}, 27(2):125--140, 1962.

\bibitem{Shepard1962-2}
Roger~N. Shepard.
\newblock The analysis of proximities: Multidimensional scaling with an unknown distance function. {II}.
\newblock {\em Psychometrika}, 27(3):219--246, 1962.

\bibitem{s-gter-26}
Jacob Steiner.
\newblock Einige {G}esetze {\"u}ber die {T}heilung der {E}bene und des {R}aumes.
\newblock {\em J. f{\"u}r die reine und angewandte {M}athematik}, 1:349--364, 1826.
\newblock \href {https://doi.org/10.1515/crll.1826.1.349} {\path{doi:10.1515/crll.1826.1.349}}.

\bibitem{DBLP:conf/icml/TeradaL14}
Yoshikazu Terada and Ulrike von Luxburg.
\newblock Local ordinal embedding.
\newblock In {\em {ICML}}, volume~32 of {\em {JMLR} Workshop and Conf. Proc.}, pages 847--855. JMLR.org, 2014.
\newblock URL: \url{http://proceedings.mlr.press/v32/terada14.html}.

\bibitem{VankadaraLHWL23}
Leena~Chennuru Vankadara, Michael Lohaus, Siavash Haghiri, Faiz~Ul Wahab, and Ulrike von Luxburg.
\newblock Insights into ordinal embedding algorithms: {A} systematic evaluation.
\newblock {\em J. Mach. Learn. Res.}, 24:191:1--191:83, 2023.
\newblock URL: \url{http://jmlr.org/papers/v24/21-1170.html}.

\bibitem{warren:68}
Hugh Warren.
\newblock Lower bounds for approximation by nonlinear manifolds.
\newblock {\em Trans. Amer. Math. Soc.}, 133:167--178, 1968.
\newblock \href {https://doi.org/10.1090/S0002-9947-1968-0226281-1} {\path{doi:10.1090/S0002-9947-1968-0226281-1}}.

\end{thebibliography}

\end{document}